\documentclass{IEEEtran}
\usepackage{amsmath, stackrel, amssymb}
\usepackage{amsthm}
\usepackage{graphicx}
\usepackage{multirow, threeparttable, makecell, booktabs}
\usepackage{bm}
\usepackage{array}
\usepackage{algorithm}
\usepackage{algpseudocode}
\newcommand{\T}{^\top}
\newtheorem{rem}{Remark}
\newtheorem{assu}{Assumption}
\newtheorem{lemma}{Lemma}
\usepackage{color}
\usepackage[colorlinks,
            linkcolor=blue,
            anchorcolor=blue,
            citecolor=blue]{hyperref}
\usepackage{cite}
\hypersetup{
	colorlinks=true,
	linkcolor=blue,
	filecolor=red,      
	urlcolor=blue,
	citecolor=blue,
}

\title{\LARGE \bf
Accurate Control under Voltage Drop for Rotor Drones}

\author{Yuhang Liu, Jindou Jia, Zihan Yang, Kexin Guo$^*$, Bin Yang, Lidan Xu, and Taihang Chen
\thanks{*This work was supported in part by the National Key Research and Development Program of China (No.2022YFB4701301), National Natural Science Foundation of China (No.62273023), Major Science and Technology Innovation Program of Hangzhou (No.2022AIZD0137), and Beijing Nova Program under Grant (No.20230484266).}
\thanks{Yuhang Liu and Zihan Yang are with the School of Aeronautic Science and Engineering, Beihang University, Beijing 100191, China. E-mail: lyhbuaa@buaa.edu.cn and snrt\_zzhan@buaa.edu.cn.}
\thanks{Jindou Jia is with the School of Automation Science and Electrical Engineering, Beihang University, Beijing 100191, China. E-mail: jdjia@buaa.edu.cn.}
\thanks{Bin Yang is with the School of Reliability and Systems Engi-
neering, Beihang University, Beijing 100191, China. E-mail: yangbin2021@buaa.edu.cn.}
\thanks{Lidan Xu and Taihang Chen are with the School of Cyber Science and Technology, Beihang University, Beijing 100191, China. E-mail: lidanxu@buaa.edu.cn and zb2039101@buaa.edu.cn.}
\thanks{Kexin Guo is with the School of Aeronautic Science and Engineering, Beihang University, Beijing 100191, China, and also with the Hangzhou Innovation institute, Beihang University, Hangzhou 310051, China. E-mail: kxguo@buaa.edu.cn. (* Corresponding author)}
}

\begin{document}

\maketitle
\begin{abstract}

This letter proposes an anti-disturbance control scheme for rotor drones to counteract voltage drop (VD) disturbance caused by voltage drop of the battery, which is a common case for long-time flight or aggressive maneuvers. Firstly, the refined dynamics of rotor drones considering VD disturbance are presented. Based on the dynamics, a voltage drop observer (VDO) is developed to accurately estimate the VD disturbance by decoupling the disturbance and state information of the drone, reducing the conservativeness of conventional disturbance observers. Subsequently, the control scheme integrates the VDO within the translational loop and a fixed-time sliding mode observer (SMO) within the rotational loop, enabling it to address force and torque disturbances caused by voltage drop of the battery. Sufficient real flight experiments are conducted to demonstrate the effectiveness of the proposed control scheme under VD disturbance.
\end{abstract}


\section{INTRODUCTION}

Rotor drones have gained widespread popularity in recent years due to their low cost, easy maintenance, and flexibility \cite{quadrotor-popular1,ICLR}. With the application of rotor drones operating in the unstructured environment, the safety of flight has received increasing attention \cite{yangbin,safe-control-Vasileios}.
While existing works have achieved satisfying results aiming at disturbances like rain disturbance \cite{TIE}, wind disturbance \cite{wind-disturbance}, model uncertainty \cite{robust_control}, center of gravity shift \cite{COG_shift,ICUAS} and ground effect interaction \cite{ground-effect}, 
few works have considered the disturbance resulted from voltage drop of the battery, i.e., voltage drop (VD) disturbance \cite{voltage_plunge1,battery_model2}. And these methods primarily rely on neural networks for nonlinear mapping, thus lacking interpretability.

Most drones are not recommended for prolonged flights or large maneuvering flights, as the chemical properties of their batteries may not support a sustained and stable current and voltage output. More specifically, due to the internal resistance of the battery under non-ideal conditions, its output voltage and current gradually decrease with the battery continues to discharge. Subsequently, the drone may suffer insufficient lift and decreased flight accuracy, leading to consequent danger. For the coaxial drone used in this letter which needs to perform long-endurance flights or payload-transport missions, the high demand of power may lead to an obvious voltage drop of the battery and the height of the drone decreases simultaneously, as illustrated in Fig. \ref{Trajectory0931}. Further, voltage drop of the battery is also a common challenge countered by other drones as executing long-time flight or aggressive maneuvers, \cite{voltage_plunge1,voltage_plunge2,voltage_plunge3}, leading to a decrease in flight stability \cite{battery-analysis}. 

Specific solutions for high-accuracy control are lacked to deal with disturbance resulted from voltage drop of the battery. In this letter, we treat the voltage drop of the battery as the VD disturbance and estimate it by a specially designed voltage drop observer (VDO), which utilizes real-time state measurement of the drone to improve the estimation performance. Besides, a fixed-time slide mode observer (SMO) is designed to estimate the torque disturbance, which is in the rotational loop and caused by voltage drop of the battery. Through experimental verification, the proposed control scheme addressing the voltage drop problem has been confirmed to be an effective and innovative method.




\begin{figure}[t]
    \centering
    
    \includegraphics[width=0.28\textwidth]{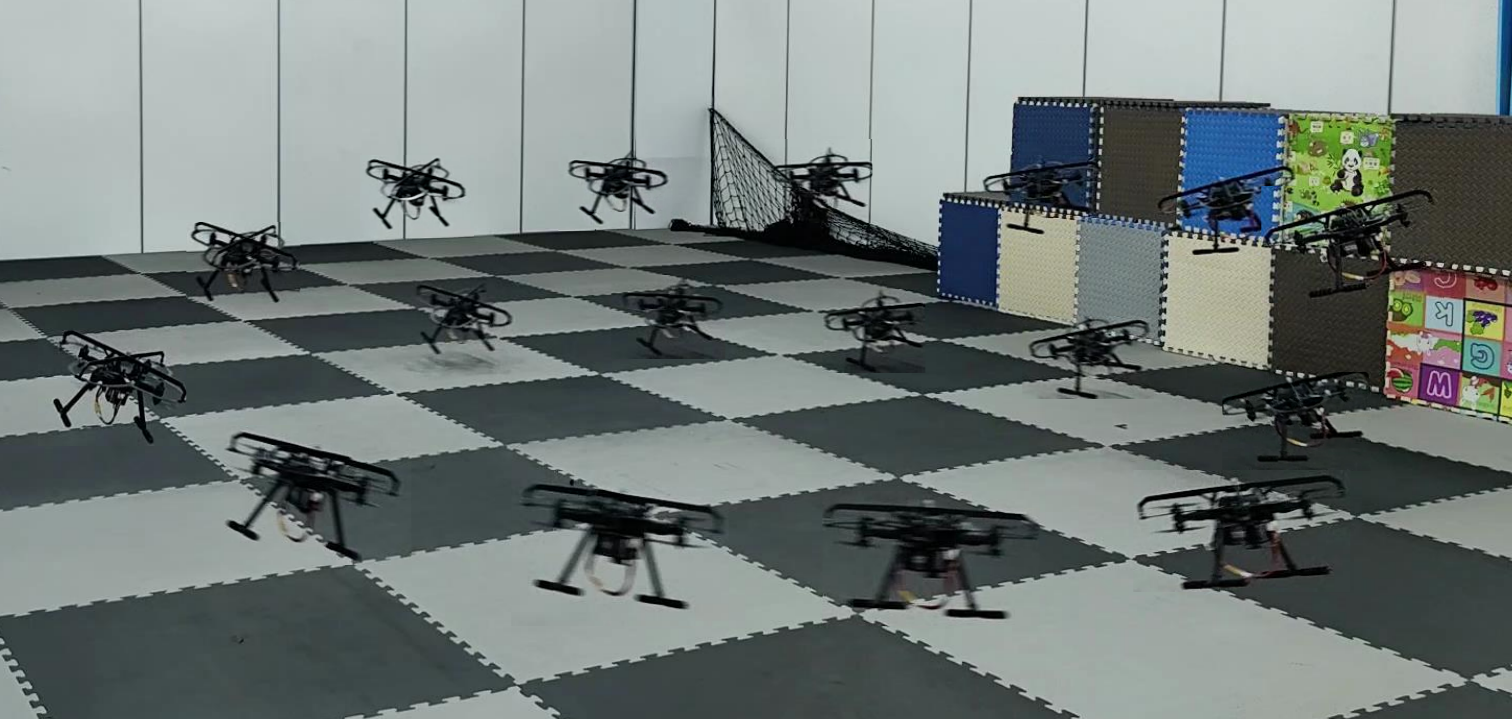}
    \hspace{-0.25cm} 
    \includegraphics[width=0.2005\textwidth]{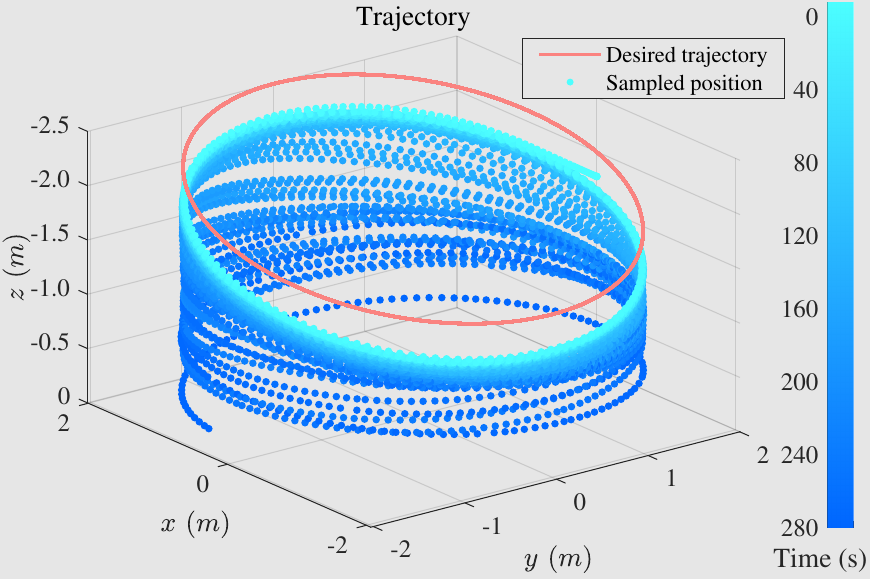}
    
    \caption{Flight scenario and trajectory of a drone conducting long-time flight. It can be seen that the drone with baseline controller descends gradually as the flight time progresses.} 
    \label{Trajectory0931}
\end{figure}
To recap, our contributions include:
\begin{itemize}
    \item A VDO based control scheme is proposed, which effectively decouples the lift loss resulting from the voltage drop of the battery and the state information of the drone. The VDO would enable the drone to conduct long-duration flights or execute aggressive maneuvers with consistent performance, maintaining stability and control accuracy despite variations in battery voltage.
    \item A rotor drone with coaxial structure of rotors is designed to increase its load capacity and mitigate the gyroscopic effect of rotors. Furthermore, the drone is capable of stable flight, ensuring reliability and adaptability for various missions. However, the increased weight from additional rotors makes the drone susceptible to the VD disturbance.
    \item The effectiveness of the proposed control scheme is verified under voltage drop of the battery compared to existing control schemes \cite{DOBC}. The experimental results demonstrate that the proposed control scheme effectively ensures the stability of altitude control for unmanned aerial vehicles during long-duration flights, even when faced with significant battery voltage reduction.
\end{itemize}

The rest of this letter is arranged as follows. Section \ref{related work} reviews related works. Section \ref{preliminaries} introduces notations and the drone model. 
Section \ref{control scheme} proposes an anti-disturbance control scheme addressing VD disturbance and the stability is analyzed. Meanwhile, the VDO is proposed in this section. The real-world experiments are detailed in Section \ref{verification}. Section \ref{conclusion} summarizes this letter.

\section{RELATED WORK}\label{related work}
In this section, some previous researches about anti-disturbance control on drones, and voltage drop of the battery of drones are listed.

\subsection{Anti-disturbance Control of Drones}

Anti-disturbance control has been developed for decades, and many efforts have been made to address different types of disturbances \cite{TIE,wind-disturbance,robust_control, ICUAS,COG_shift,ground-effect}. 
For example, \cite{robust_control} proposes a robust control scheme to deal with model uncertainties and disturbances for rotor drones. However, the robust control method usually adopt the norm-bounded constraints on time-varying disturbances, exhibiting conservative characteristics. Subsequently, control methods based on disturbance observers have been developed, demonstrating reduced conservativeness. 
\cite{wind-disturbance} introduces a frequency-based wind gust estimation method using a nonlinear disturbance observer (NDO), achieving higher accuracy by considering gust frequency. However, the effectiveness of the proposed method is not validated through closed-loop experiments. \cite{COG_shift} proposes an aerodynamic drag model that considers variations in the drag surface, effectively decoupling the state information of the drone and external disturbances. Based on the wind model, a disturbance observer is designed to effectively suppress the influence of aerodynamic drag. Inspired by these related works, the work in this letter aims to decouple the VD disturbance and the state information of the drone, enabling more accurate disturbance estimation and subsequent compensation.

\subsection{Voltage Drop of the Battery of Drones}

Due to the open-loop control structure of electronic speed controller (ESC) in most drone applications, it is hard to guarantee that the actual rotation speed of each rotor is consistent with its desired speed as the result of voltage drop of the battery, and other elements. However, few works reported the impact of voltage drop on rotor drones and provide a reasonable and feasible solution. \cite{voltage_plunge1} randomizes the mapping relationship between the desired thrust and the desired rotation speed of each rotor to help drones learn how to deal with voltage drop of the battery. However, the learning process requires extensive data and time, which is laborious. \cite{battery_model2} simulates the battery voltage using a gray-box battery model \cite{battery_model1}, which is also laborious for all involved quantities should have been identified from extensive data. With respect to this issue, the VDO is proposed in this letter to estimate the lift loss produced by the voltage drop of the battery, decoupling the states of the drone and the lift loss, and reducing the conservativeness of conventional NDOs.

\section{PRELIMINARIES}\label{preliminaries}
This section mainly introduce the used notations and the model of the coaxial octocopter drone.
\subsection{Notations} 

In this letter, the transpose of $*$ is denoted as $\left( * \right)\T$, and $\mathbb{R}^{m \times n}$ represents an $m \times n$-dimensional real space. The notations $s_*$ and $c_*$ correspond to $ \sin (*) $ and $ \cos (*) $, respectively. The desired value of $*$ is indicated as ${\left(  *  \right)_d}$. The largest and smallest eigenvalues of a given matrix are represented by ${\lambda _M}\left( * \right)$ and ${\lambda _m}\left( * \right)$, respectively. Additionally, $\hat{*}$ denotes the estimation of $*$, while $\dot{*}$ and $\ddot{*}$ represent the first and second order time derivatives of $*$, respectively. The skew-symmetric matrix associated with vector $*$ is expressed as $*^{\times}$. The Euclidean norm of a vector or the Frobenius norm of a matrix is denoted as $\left\| * \right\|$, whereas the Manhattan norm of a vector is given by $\left\| * \right\|_1$.

As a convention, the skew-symmetric operator of a vector $\boldsymbol{x} = {\left[{\begin{array}{*{20}{c}}{{x_1}}&{{x_2}}&{{x_3}}\end{array}} \right]\T}$ is defined as 
	\begin{align*}
	{\boldsymbol{x}^ \times } = \left[ {\begin{array}{*{20}{c}}
		0&{ - {x_3}}&{{x_2}}\\
		{{x_3}}&0&{ - {x_1}}\\
		{ - {x_2}}&{{x_1}}&0
		\end{array}} \right].   
	\end{align*}

\subsection{Coaxial Octocopter Drone Model}

\begin{figure}
		\centering
		\includegraphics[width=0.45\textwidth]{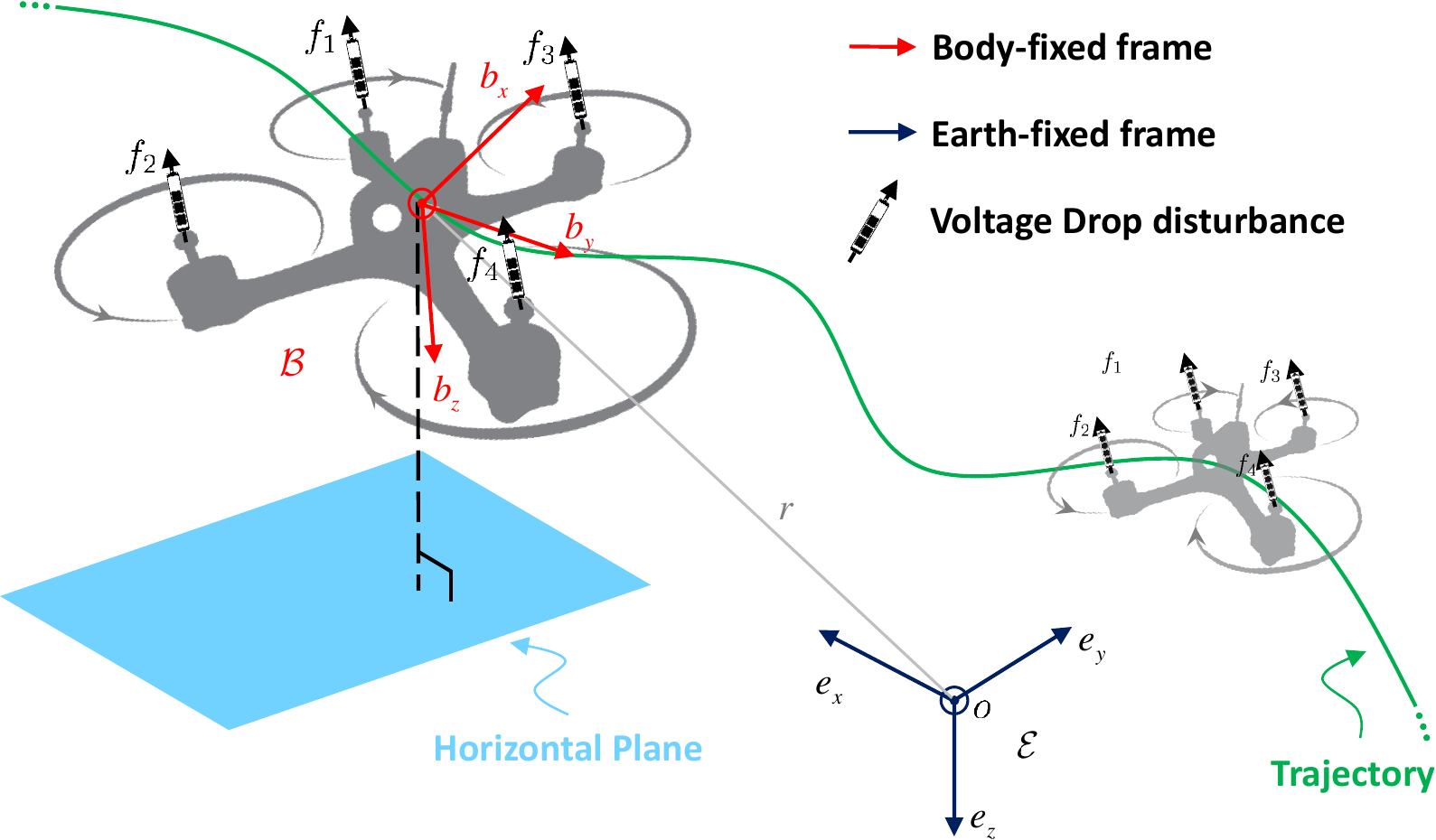}
		\caption{Schematic of a rotor drone. Two frames are defined: the earth-fixed frame (North-East-Down) $\mathcal {E}$ and the body-fixed frame $\mathcal {B}$. Disturbance caused by voltage drop in the battery is primarily considered in this letter.}
		\label{Coordinate System}
	\end{figure}
    
As illustrated in Fig. \ref{Coordinate System}, states of the drone are described in two coordinate frames: earth-fixed frame (North-East-Down) $\mathcal {E}={\left[{\begin{array}{*{20}{c}}{{\boldsymbol e_x}}&{{\boldsymbol e_y}}&{{\boldsymbol e_z}}\end{array}} \right]}$ and body-fixed frame $\mathcal {B}={\left[{\begin{array}{*{20}{c}}{{\boldsymbol b_x}}&{{\boldsymbol b_y}}&{{\boldsymbol b_z}}\end{array}} \right]}$. 
To enhance clarity, ${\left(  \cdot  \right)^E}$ and ${\left(  \cdot  \right)^B}$ are the representation of the physical variables in $\mathcal {E}$ and $\mathcal {B}$, respectively. $\boldsymbol \eta^B ={{\left[ \begin{matrix}\phi  & \theta  & \psi \end{matrix} \right]}\T}$ is defined as Euler angles of the drone and ${{\boldsymbol \omega }^B}=\left[ \begin{matrix}   p & q & r \end{matrix} \right]\T$ represents the angular velocity of the drone in $\mathcal {B}$ which can be obtained by ${{\left[ \begin{matrix} {\dot{\phi }} & {\dot{\theta }} & {\dot{\psi }} \end{matrix} \right]}\T}=\boldsymbol C {{\boldsymbol \omega }^B}$, where $\boldsymbol C$ is defined in \cite{RAL-C}. 

The kinematics and dynamics of the coaxial octocopter drones can be modeled as follows,
    \begin{equation}
      \begin{cases}
          {\dot {\boldsymbol p}^E=\boldsymbol v^E},\\
          {\dot {\boldsymbol{R}} = \boldsymbol{R}{\boldsymbol {\omega}}^ {B\times}},
      \end{cases}
       \end{equation}
    \begin{equation}\label{dynamics}
      \begin{cases}
          m\boldsymbol{a}^E =  \boldsymbol {F}^E + {\boldsymbol G}^E + \Delta {\boldsymbol F}^E,\\
          {\boldsymbol{J} {\boldsymbol{\dot \omega}^B} =  - {\boldsymbol{\omega}^B}^\times \boldsymbol{J\omega}^B  + \boldsymbol \tau^B + \boldsymbol{\tau}_{dis}^B},
      \end{cases}
       \end{equation}
where $\boldsymbol p^E \in\mathbb{R}^{3}$, $\boldsymbol v^E \in\mathbb{R}^{3}$, and $\boldsymbol a^E \in\mathbb{R}^{3}$ represent the position, velocity, acceleration of the drone in $\mathcal{E}$ respectively. $\boldsymbol R\in \mathbb{R}^{3\times 3}$ is a rotation matrix from $\mathcal {B}$ to $\mathcal {E}$. $m$ is the mass of the drone, $\boldsymbol {F}^E=-f{\boldsymbol{b}_z}$ represents the total thrust of the drone, $\boldsymbol G^E=mg{\boldsymbol{e}_z}$ represents the gravity of the drone. $\Delta {\boldsymbol F}^E$ represents the lift loss of the drone resulted from voltage drop of the battery. $\boldsymbol J \in \mathbb{R}^{3 \times 3}$ is the inertia matrix in $\mathcal{B}$. Control input is defined as $\boldsymbol U=\left[{\begin{array}{*{20}{c}}{{f}}&{(\boldsymbol {\tau}^{B})\T}\end{array}} \right]\T$, where $f$ and $\boldsymbol {\tau}^{B}$ represent the total thrust and torque generated by all the rotors mounted on the drone, respectively. $\boldsymbol {\tau}_{dis}^B\in\mathbb{R}^{3}$ represents the moment generated by voltage drop of the battery. 

\section{CONTROL FRAMEWORK for VOLTAGE DROP}\label{control scheme}
\begin{figure}
    \centering
    \includegraphics[width=0.45\textwidth]{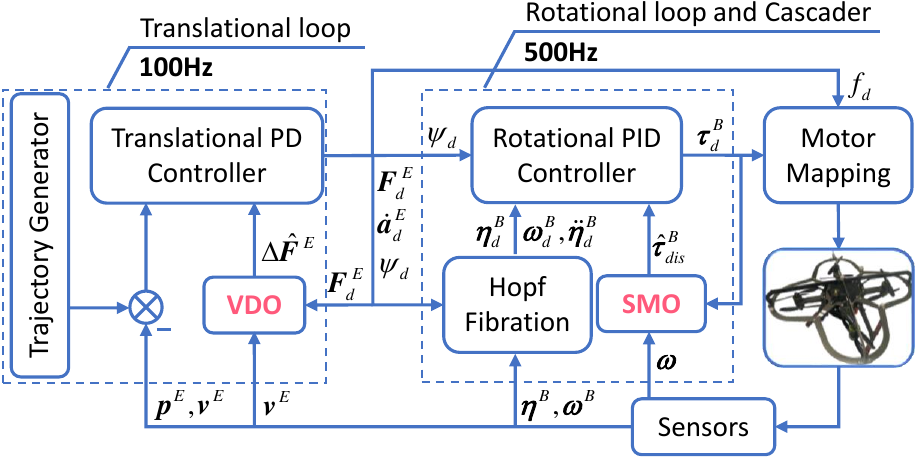}
    \caption{Control architecture of the anti-disturbance control strategy addressing VD disturbance. Two observers estimating the corresponding disturbance are designed (i.e., VDO and SMO) and embedded in the control architecture. The translational loop operates at a frequency of 100 Hz while the rotational loop  operates at a frequency of 500 Hz.}
    \label{Control scheme}
\end{figure}
This section presents an anti-disturbance control framework with stability analysis, as shown in Fig. \ref{Control scheme}. The proposed control strategy utilizes a cascade control structure \cite{cascade_control2}
, consisting of a translational loop for position control and a rotational loop for attitude control. Additionally, the VDO and SMO are introduced to address the VD disturbance. 

\subsection{Translational Loop}
\subsubsection{Baseline Controller}
The translational loop is designed to track the desired position $\boldsymbol p_d^E = {\left[ {\begin{array}{*{20}{c}}{{x_d}}&{{y_d}}&{{z_d}}\end{array}} \right]\T}$ and generate the corresponding control force $\boldsymbol F_d^E$. To quantify the tracking errors, $\boldsymbol e_{p}^{E}=\boldsymbol p_{d}^{E}-{\boldsymbol {p}^{E}}$ and $\boldsymbol e_{v}^{E}=\boldsymbol v_{d}^{E}-{\boldsymbol {v}^{E}}$ are defined to present the position tracking error and the velocity tracking error, respectively.  The baseline controller of the translational loop is designed as follows.
\begin{align} \label{translational baseline}
    \begin{cases}
        \boldsymbol a_d^E=\boldsymbol K_p  \boldsymbol e_p^E + \boldsymbol K_v \boldsymbol e_v^E -g\boldsymbol e_z + \ddot {\boldsymbol p}_d^E,\\
        \boldsymbol F_d^E=m\boldsymbol a_d^E-\Delta {\hat {\boldsymbol F}^E},
    \end{cases}
\end{align}
where $\boldsymbol K_p \in\mathbb R^{3 \times 3}$ and $\boldsymbol K_v\in\mathbb R^{3 \times 3}$ are gains of the translational controller, both of which possess a positive diagonal structure, $\Delta {\hat {\boldsymbol F}^E} \in\mathbb R^{3}$ is defined as the VD disturbance estimated by the VDO.  

\subsubsection{Voltage Drop observer}

During flight missions, particularly in long-time flight or aggressive maneuvers, continuous supply of high current leads to an increase in battery temperature, which rises the internal resistance of the battery, resulting in the voltage drop problem. To address the VD disturbance for the drone, a VDO is designed in this letter.

Before procedure, $\Delta \boldsymbol F^E$ can be further detailed as
\begin{align}\label{delta F}
\Delta \boldsymbol F^E=\Delta f\cdot \underbrace{\left[ \begin{matrix}
   {{c}_{\psi }}{{s}_{\theta }}{{c}_{\phi }}+{{s}_{\psi }}{{s}_{\phi }}  \\
   {{s}_{\psi }}{{s}_{\theta }}{{c}_{\phi }}-{{c}_{\psi }}{{s}_{\phi }}  \\
   {{c}_{\theta }}{{c}_{\phi }} \\
\end{matrix} \right]}_{\boldsymbol \Theta},
\end{align}
where $\Delta f$ is a scalar representing the magnitude of thrust loss, which is caused by the voltage drop of the battery and other elements, $\boldsymbol \Theta\in \mathbb R^{3 \times 1}$ is related with the state information of the drone.

\begin{assu}\label{Assumption1}
    The first-order time derivative of $\Delta f$ is bounded, i.e., $\|\Delta \dot{f}\| \le \mu$, where $\mu$ is a constant.
\end{assu}

\begin{rem}
     $\Delta {f}$ is resulted from voltage drop of the battery and its first-order time derivative can be considered to be norm bounded, since the voltage of the battery cannot change abruptly like the step signal. Further, the rotor speed, which directly affects the produced lift, cannot change abruptly due to rotational inertia of rotors. In this case, Assumption \ref{Assumption1} can be deemed reasonable.
\end{rem}
It is clear that the VD disturbance can be explicitly separated into variation in lift value and a vector composed of euler angles of the drone, i.e., $\Delta f$ and  $\boldsymbol \Theta$. In this case, the VD disturbance can be decoupled from the state information of the drone.

Based on \eqref{delta F}, the structure of the VDO is designed as
\begin{align}\label{VDO}
   \begin{cases}
      \dot{\wp}=-\boldsymbol \zeta(\boldsymbol G^E+\boldsymbol F^E+\boldsymbol \Theta \Delta \hat{f}),  \\
       \Delta \hat{f}=\wp+\boldsymbol \zeta m{{\boldsymbol v}^{E}},  \\
   \end{cases}
\end{align}
where $\wp$ is an auxiliary variable, $\boldsymbol \zeta \in\mathbb R^{1 \times 3}$ is the gain of VDO which is required to be adjusted.

\begin{lemma}
     Considering system \eqref{dynamics}, the proposed VDO in \eqref{VDO} is stable if $\boldsymbol {\zeta \Theta}$ is positive definite.
\end{lemma}
\begin{proof}
    We define $\Delta \tilde{f} = \Delta \hat{f} - \Delta f$ as the estimation error of the $\Delta f$, and it can be obtained that
\begin{align}\label{stability of VDO}
   \Delta \dot{\tilde{f}}&=\Delta \dot{\hat{f}}-\Delta \dot{f} \notag \\
   &=\dot {\boldsymbol \wp}+ \boldsymbol \zeta m \boldsymbol \dot{v}^E - \Delta \dot{f} \notag \\
   &=-\boldsymbol \zeta(\boldsymbol F^E+\boldsymbol G^E+\boldsymbol \Theta \Delta \hat{f})+\boldsymbol \zeta(\boldsymbol G^E \notag\\
   &\quad\,  + \boldsymbol F^E + \boldsymbol\Theta \Delta f)-\Delta \dot{f}\notag \\
   &=-\boldsymbol {\zeta\Theta} (\Delta \hat{f}-\Delta f)-\Delta \dot{f} \notag\\ 
   &=-\boldsymbol {\zeta\Theta} \Delta {\tilde{f}}-\Delta \dot{f},
\end{align}
where $\Delta \dot{f}$ is norm bounded, according to Assumption \ref{Assumption1}. 

It can be verified that the VDO is stable if $\boldsymbol {\zeta \Theta}$ is positive definite and the estimation error can be limited within a range related to $\Delta \dot{f}$. 
\end{proof}


\begin{rem}
    Compared with the NDO \cite{DOBC}, the conservativeness of the VDO is reduced to a large extent by utilizing the knowledge of $\boldsymbol \Theta$. And it can be applied to drones conducting long-time flight or aggressive maneuvers. 
\end{rem}

\begin{rem}
    The drone may undergo aggressive maneuvers and $\boldsymbol{\Theta}$ is associated with the time-varying attitude angles, making it possible for $\boldsymbol {\zeta \Theta}$ to be negative if $\boldsymbol{\zeta}$ is a constant. In this case, $\boldsymbol{\zeta}$ can be designed as a state-dependent variable, such as $\boldsymbol{\Theta}\T$.
\end{rem}

 \subsection{Stability Analysis of Translational Loop}
 
The combination of baseline controller and VDO can be proven to be stable mathematically. By resorting to \eqref{translational baseline}, we can obtain the derivatives of $\boldsymbol e_p^E$ and $\boldsymbol e_v^E$ as
\begin{align} \label{state_function}
    \underbrace{\left( \begin{matrix}
   {{{\dot{\boldsymbol e}}}_{p}}  \\
   {{{\dot{\boldsymbol e}}}_{v}}  \\
\end{matrix} \right)}_{\dot{\boldsymbol e}}=\underbrace{\left( \begin{matrix}
   {{\boldsymbol 0}_{3\times 3}} & {{\boldsymbol I}_{3\times 3}}  \\
   \frac{-{{\boldsymbol K}_{p}}}{m} & \frac{-{{\boldsymbol K}_{v}}}{m}  \\
\end{matrix} \right)}_{\boldsymbol A} \underbrace{\left( \begin{matrix}
   {{\boldsymbol e}_{p}}  \\
   {{\boldsymbol e}_{v}}  \\
\end{matrix} \right)}_{\boldsymbol e}+\underbrace{\left( \begin{matrix}
   {{\boldsymbol 0}_{3\times 3}} & {{\boldsymbol 0}_{3\times 3}}  \\
   {{\boldsymbol 0}_{3\times 3}} & \frac{{{\boldsymbol I}_{3\times 3}}}{m}  \\
\end{matrix} \right)}_{\boldsymbol B} \underbrace{\left( \begin{matrix}
   {{\boldsymbol 0}_{3\times 1}}  \\
   {\Delta {\hat {\boldsymbol F}^E}}  \\
\end{matrix} \right)}_{\tilde {\boldsymbol d}},
\end{align}
where ${\Delta{{\tilde{\boldsymbol F}}}^E}=\Delta\hat{\boldsymbol F}^E-\Delta\boldsymbol F^E=\boldsymbol \Theta {\Delta \tilde f}$.

It can be checked that $\boldsymbol A$ has negative definite structure if $\boldsymbol K_p$ and $\boldsymbol K_v$ have positive definite structures. As a consequence, there must exist a positive definite symmetric matrix $\boldsymbol M$ that meets the equation $\boldsymbol A\T \boldsymbol M+\boldsymbol M \boldsymbol A=-\boldsymbol I$. Subsequently, a Lyapunov function is designed as follows,
\begin{align} \label{lyapunov}
     V = {\boldsymbol e\T} \boldsymbol{Me} + \frac{1}{2}({{\Delta \tilde f}})\T{\Delta \tilde f}.
\end{align}

Differentiate $V$, it can be implied that 
\begin{align}\label{V_dot}
  \dot{ V} =&\dot{\boldsymbol e}\T \boldsymbol M\boldsymbol e+{\boldsymbol e}\T \boldsymbol M \dot{\boldsymbol e}+({\Delta \tilde f})\T\Delta \dot{\tilde{f}}\notag \\
  = &{\boldsymbol e\T}{\boldsymbol A\T}{\boldsymbol M}{\boldsymbol e}+{\boldsymbol {\tilde d}\T}{\boldsymbol B\T}{\boldsymbol M}{\boldsymbol e}+{\boldsymbol e\T}{\boldsymbol M}{\boldsymbol A}{\boldsymbol e}\notag \\
  &+{\boldsymbol e\T}{\boldsymbol M}{\boldsymbol B}{\boldsymbol {\tilde d}}+ ({\Delta \tilde f})\T\Delta \dot{\tilde{f}} \notag \\
 =& {\boldsymbol e\T}\left( {{\boldsymbol A\T}\boldsymbol M+\boldsymbol{MA}} \right)\boldsymbol e + 2{\boldsymbol e\T}\boldsymbol{MB}\tilde {\boldsymbol d} +
 ({\Delta \tilde f})\T\Delta \dot{\tilde{f}} \notag\\
 =&  - {\boldsymbol e\T}\boldsymbol e + 2{\boldsymbol e\T}\boldsymbol{MB}\tilde {\boldsymbol d} + ({\Delta \tilde f})\T\Delta \dot{\tilde{f}}.
\end{align}

Furthermore, by resorting to Young's inequality \cite{young's_inequality}, it can be verified that 
\begin{subequations}\label{inequalities}
    \begin{align}
        -{{\boldsymbol e}\T}\boldsymbol e &\le -\frac{{{\boldsymbol e}\T}\boldsymbol {Me}}{{{\lambda }_{M}}\left( \boldsymbol M \right)}, \\
        -({\Delta \tilde f})\T\Delta \dot{\tilde{f}} &\le \frac{1}{4}({\Delta \tilde f})\T{{\Delta \tilde f}}+\delta, \\
        {{\boldsymbol e}\T}\boldsymbol {MB}\tilde{\boldsymbol d}&\le \frac{1}{2\varepsilon }\lambda _{M}^{2}\left( \boldsymbol {MB} \right){{\boldsymbol e}\T}\boldsymbol e+\frac{\varepsilon }{2}{{\tilde{\boldsymbol d}}\T}\tilde{\boldsymbol d},
    \end{align}
\end{subequations}
where $\varepsilon$ is an arbitrary positive constant and ${\delta}=(\Delta \dot{ \tilde f})\T{\Delta \dot{ \tilde f}}$ is a bounded value related with $\Delta \dot{ \tilde f}$, which is norm bounded according to Assumption \ref{Assumption1}.

From \eqref{stability of VDO}, it can be obtained that 
\begin{subequations}\label{inequalities2}
    \begin{align}
   (\Delta \tilde{f})\T{\Delta {{\dot{\tilde{f}}}}}&=-\boldsymbol {\zeta\Theta} (\Delta \tilde{f})\T{\Delta \tilde{f}}-(\Delta \tilde{f})\T{\Delta {{\dot{\tilde{f}}}}}\notag \\ 
 & \le -{{\lambda }_{m}}(\boldsymbol {\zeta\Theta})(\Delta \tilde{f})\T{\Delta \tilde{f}}+\frac{1}{4}(\Delta \tilde{f})\T{\Delta \tilde{f}}+\delta,\\
  \label{inequalities2-b} {{{\tilde{\boldsymbol d}}}\T}\tilde{\boldsymbol d}&=(\Delta \tilde{f})\T{{\boldsymbol \Theta }\T}\boldsymbol \Theta {\Delta \tilde{f}} \le {{\lambda }_{M}}({{\boldsymbol \Theta }\T}\boldsymbol {\Theta} )(\Delta \tilde{f})\T{\Delta \tilde{f}}.
 \end{align} 
\end{subequations}

Combining \eqref{V_dot} - \eqref{inequalities2}, \eqref{V_dot} can be adjusted that
\begin{align}
\dot{ V} \le &-{{\boldsymbol e}\T}\boldsymbol e+\frac{1}{ \varepsilon }\lambda _{M}^{2}\left( \boldsymbol {MB} \right){{\boldsymbol e}\T}\boldsymbol e+ \varepsilon {{{\tilde{\boldsymbol d}}}\T}\tilde{\boldsymbol d}+(\Delta \tilde{f})\T \Delta\dot {\tilde{f}} \notag\\
 \le &-\frac{ \varepsilon -\lambda_{M}^{2}\left( \boldsymbol {MB} \right)}{ \varepsilon {{\lambda }_{M}}\left( \boldsymbol M \right)}{{\boldsymbol e}\T}\boldsymbol {Me}+ \varepsilon {{\lambda }_{M}}({{\boldsymbol \Theta }\T}\boldsymbol {\Theta} )(\Delta \tilde{f})\T{\Delta \tilde{f}}\notag\\
&-{{\lambda }_{m}}(\boldsymbol {\zeta\Theta})(\Delta \tilde{f})\T{\Delta \tilde{f}}+\frac{1}{4}(\Delta \tilde{f})\T{\Delta \tilde{f}}+\delta \notag\\
  \le&-{{\sigma }_{1}}{{\boldsymbol e}\T}\boldsymbol {Me}-{{\sigma }_{2}}(\Delta \tilde{f})\T{\Delta \tilde{f}}+\delta,
\end{align}
where ${{\sigma }_{1}}=\frac{\varepsilon -\lambda _{M}^{2}\left( \boldsymbol {MB} \right)}{\varepsilon {{\lambda }_{M}}\left( \boldsymbol M \right)}$ and ${{\sigma }_{2}}={{\lambda }_{m}}(\boldsymbol {\zeta\Theta} )-\varepsilon {{\lambda }_{M}}({{\boldsymbol \Theta }\T}\boldsymbol \Theta )-\frac{1}{4}$. $\sigma_1$ can be ensured to be positive with $\varepsilon -\lambda _{M}^{2}\left( \boldsymbol {MB} \right)>0$ satisfied. $\sigma_2$ can be guaranteed to be positive by adjusting $\boldsymbol \zeta$ properly.
We can obtain that $\dot{V} \le -\gamma V+\delta$ where $\gamma =\min \left\{ {{\sigma }_{1}},2{{\sigma }_{2}} \right\}$, which implies that $0\le V\left( t \right)\le {{e}^{-\gamma t}}\left[ V(0)-\frac{\delta }{\gamma } \right]+\frac{\delta }{\gamma }$.

Finally, it can be concluded that all signals of the translational loop are globally uniformly bounded. Additionally, by appropriately adjusting $\gamma$, the estimation error of VDO can converge to an arbitrarily small residual set.

\subsection{Rotational Loop}

By utilizing the Hopf Fibration \cite{Hopf_Fibration}, the desired Euler angles ${{\boldsymbol \eta }_{d}^B}$ and desired angular velocities ${\boldsymbol \omega }_d^{B}$ can be obtained. Based on these desired attitude signals, the rotational loop is then designed to ensure accurate tracking and generate the control torque $\boldsymbol \tau_d^B$.  

\subsubsection{Baseline Controller}

To prevent trajectory tracking accuracy from deteriorating due to the limited performance of the rotational loop, a rotational control strategy is designed for the rotor drone. This strategy consists of two key components: a rotational baseline controller and a fixed-time SMO. The deviation between the actual Euler angles and the reference signals is expressed as ${{\boldsymbol e}_{ \eta }}={{\boldsymbol \eta }_{d}^B}-\boldsymbol \eta^B$, while the desired angular velocity is defined as ${\boldsymbol q}_d={\boldsymbol \omega}_d^B+\boldsymbol C^{-1} \boldsymbol K_{\eta}{\boldsymbol e}_{ \eta}$. Here, $\boldsymbol K_{\eta}\in \mathbb R^{3 \times 3}$ represents a gain matrix with positive diagonal structure. Additionally, the deviation of actual angular velocities from referenced signals is given by $\boldsymbol e_{q}={\boldsymbol q}_d-\boldsymbol \omega^B$.


The rotational baseline controller is designed as
\begin{align}
    \begin{cases}
        \boldsymbol \alpha _{d}^{B}={{\boldsymbol K}_{pp}}{{\boldsymbol e}_{ q }}+{{\boldsymbol K}_{ii}}\int\limits_{{{t}_{0}}}^{t}{{{\boldsymbol e}_{q}}}d\tau,
        \\
        \boldsymbol \tau _{d}^{B}=\boldsymbol J \boldsymbol \alpha _{d}^{B}-\boldsymbol {\hat {\tau}} _{dis}^{B}+ {\boldsymbol{\omega}^B}^\times \boldsymbol{J\omega}^B,  \\
    \end{cases}
\end{align}
where $\boldsymbol \alpha_d^B$ represents the angular acceleration of the drone, while $\boldsymbol K_{pp} \in \mathbb R^{3 \times 3}$ and $\boldsymbol K_{ii} \in \mathbb R^{3 \times 3}$ 
are gains of rotational baseline controller, both possessing a positive diagonal structure. The time required for the SMO to converge is represented by $t_0$, and $\boldsymbol {\hat {\tau}} _{dis}^{B}$ represents the estimated torque disturbance obtained through the SMO.

\subsubsection{Fixed-Time SMO}
Compared with the translational disturbance caused by the drop in battery voltage, the corresponding torque disturbance ranges on a small scale. Since the torque disturbance caused by drop in battery voltage is primarily resulted from the differences in characteristics between motors. Then, the torque distubance satisfies $\|{\boldsymbol{\tau}}_{dis}^B\|<\varepsilon$, where $\varepsilon$ is a small positive value. Consequently, the chattering problem that commonly follows the sliding mode controller can be weakened. The structure of the SMO is designed as
\begin{align}
\left\{ {\begin{array}{*{20}{l}}
{\dot {\boldsymbol \mu}  = -(\boldsymbol{\omega}^B)^{\times} \boldsymbol J \boldsymbol{\omega}^B + {\boldsymbol J^{-1}}{\boldsymbol \tau ^B} + {\boldsymbol \xi _1}},\\
{{\boldsymbol \xi_1} =  - {l_1}\frac{{{\boldsymbol e_1}}}{{{{\left\| {{\boldsymbol e_1}} \right\|}^{{1 \mathord{\left/
 {\vphantom {1 2}} \right.
 \kern-\nulldelimiterspace} 2}}}}} - {l_2}{\boldsymbol e_1}\left\| {{\boldsymbol e_1}} \right\| + {\boldsymbol \xi_2}},\\
{{{\dot {\boldsymbol \xi} }_2} =  - {l_3}\frac{{{\boldsymbol e_1}}}{{\left\| {{\boldsymbol e_1}} \right\|}}},\\
{{{\hat {\boldsymbol \tau} }_{dis}^B} = \boldsymbol J{\boldsymbol \xi _2}},
\end{array}} \right.
\end{align}
where ${{\boldsymbol e}_{1}}=\boldsymbol \mu -\boldsymbol \omega^B$, $l_1$, $l_2$ and $l_3$ are gains need to design.
Stability analysis of the SMO can be referenced in \cite{safety_Contro_wenyu} and \cite{fixed_time_second_order_sliding_mode_control}. The performance of the SMO has been validated in our previous work \cite{safety_Contro_wenyu}.

\section{PERFORMANCE VERIFICATION}\label{verification}
This section presents several real-world experiments to validate the effectiveness of the proposed control scheme. It begins with an introduction to the experimental platform, followed by experiments conducted consecutively. Finally, a quantitative analysis is presented, which includes the root mean square error (RMSE) and mean absolute error (MAE) analyses for position tracking.

\subsection{Flight Experiment Setup}

\subsubsection{Platform Setup}
\begin{figure}
    \centering
    \includegraphics[width=0.45\textwidth]{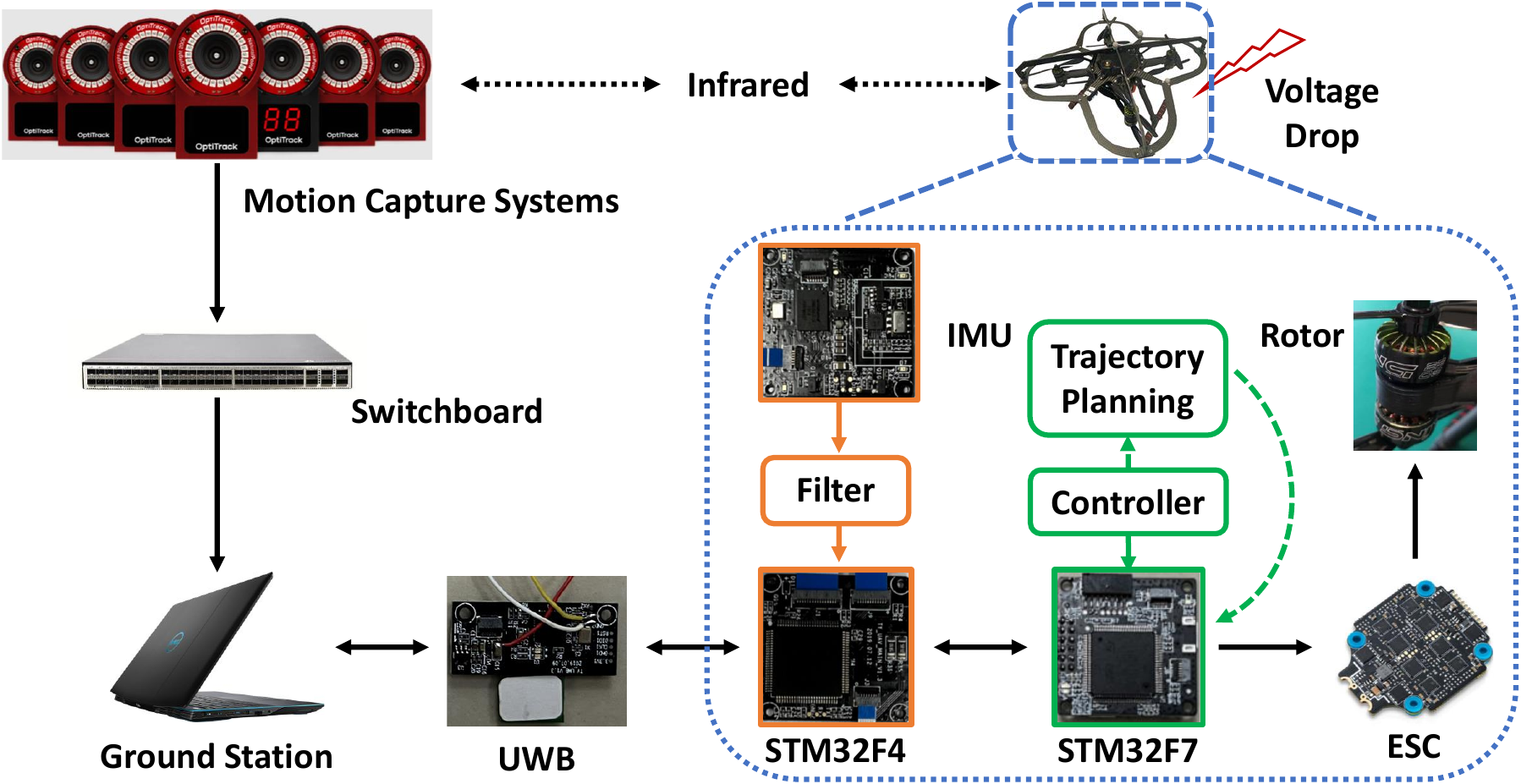}
    \caption{Hardware platform. The positioning is obtained from a motion capture system. The attitude information is obtained from the fusion of motion capture system data and the onboard IMU data, which can be regarded as having negligible noise.}
    \label{Hardware platform}
\end{figure}

As shown in Fig. \ref{Hardware platform}, the drone features four sets of propulsion units arranged in a coaxial structure to enhance its load capacity. Furthermore, the counter-rotating propellers of equal size can mitigate the gyroscopic effect. The control algorithm runs on an STM32F7, while navigation data is processed on an STM32F4, enhancing the computational capability of the onboard processors. With these specialized treatments, the drone can achieve long-duration stable flight and maintain excellent performance under VD disturbance.



\subsection{Flight Experiment Arrangement}

In an attempt to verify the effect of the proposed control scheme, indoor experiments are carried out without additional disturbance applied artificially. 
The drone is commanded to fly along a trajectory formulated as $\boldsymbol p_d^E=[2sin(\frac{2\pi}{T}t),2cos(\frac{2\pi}{T}t),-2.0-0.5sin(\frac{2\pi}{T}t)]\T m$, where $T$ represents the trajectory period and determines flight speed of the drone. The desired height indicates the distance between the drone and the ground plane. The drone is commanded to fly for more than 280 seconds so that the effect of voltage drop of the battery can be observed distinctly.

\begin{table}
    \centering
    \caption{Controller Parameters}
    \renewcommand\arraystretch{1.5}
    \begin{tabular}{cccc}
    \toprule
         Parameter & Value & Parameter & Value \\
    \midrule
         $\boldsymbol{K}_p$ & diag(31.36,31.36,7.84) & $\boldsymbol{\zeta}$ & $\boldsymbol{\Theta}\T$\\ 
         $\boldsymbol{K}_v$ & diag(11.2,11.2,5.6) & $l_1$ & 5\\
         $\boldsymbol{K}_{pp}$ & diag(32,32,24) & $l_2$  & 10\\ 
         $\boldsymbol{K}_{ii}$ & diag(5,5,5) & $l_3$ & 5\\
    \bottomrule
    \end{tabular}
    \label{controller parameters}
\end{table}

\subsection{Flight Experiment Results} 


\begin{figure}
    \centering
    \includegraphics[width=0.45\textwidth]{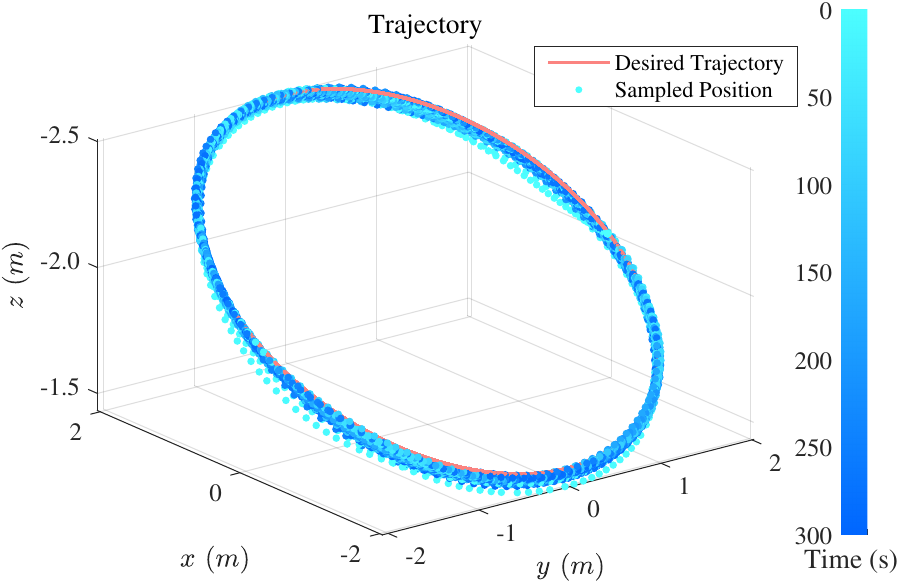}
    \caption{Trajectory of the drone with voltage drop of battery over 280 seconds duration (VDO based control scheme). As the flight time progresses, the drone holds its height successfully. 
}
    \label{VDO-effect0942}
\end{figure}

\begin{figure}
    \centering
    \includegraphics[width=0.45\textwidth]{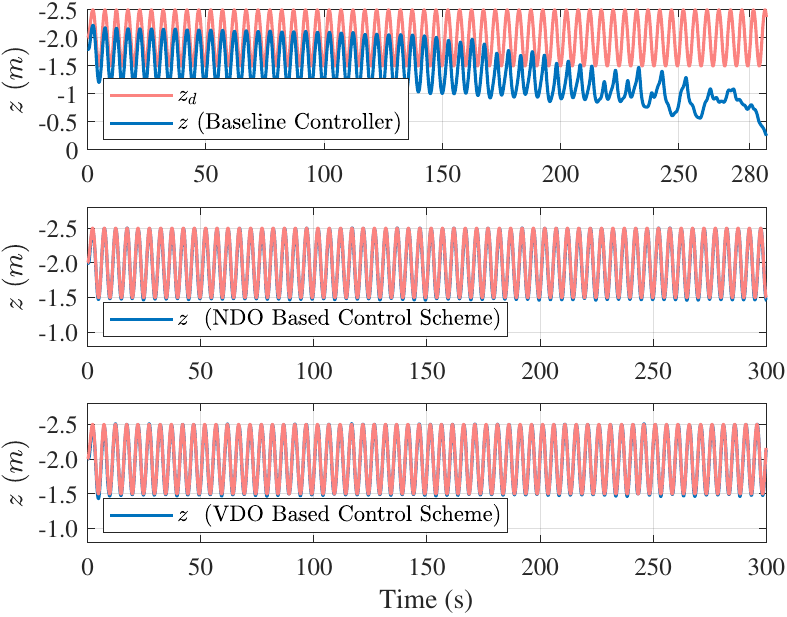}
    \caption{Position tracking performance of the drone with different control schemes. The drone with VDO or NDO based control scheme holds its height instead of descending over time successfully.}
    \label{VDO-effect0939}
\end{figure}

As illustrated in the first subplot of Fig. \ref{VDO-effect0939}, it is clear that the drone with baseline controller will descend gradually due to the voltage drop of the battery. The longer the the flight time, the more noticeable the descend phenomenon becomes. However, with the VDO enabled, the drone hold its desired height despite the voltage drop of the battery, as shown in Fig. \ref{VDO-effect0942} and the third subplot of Fig. \ref{VDO-effect0939}. Further, the VDO performs better than the NDO based control scheme in the translational loop, as illustrated in Table \ref{VDO-effect}. The initial position error illustrated in the first subplot of Fig. \ref{VDO-effect0939} is mainly caused by the lift loss resulted from coaxial rotors. 

The RMSE and the MAE are defined as
\begin{align}
    r_1&=\frac{1}{{\sqrt n }}\left\| {*_d - {*}} \right\|,\\
    r_2&=\frac{1}{n}{{\left\| {*_d - {*}} \right\|}_1},
\end{align}
where $n$ represents the sampling numbers of $*$ during flight. 

As illustrated in Table \ref{VDO-effect}, the control accuracy of VDO based scheme achieves 8.08\% and 13.22\% improvement in RMSE and MAE compared with those of the NDO based control scheme, respectively. Furthermore, compared with the baseline controller, the VDO based controller exhibits a more pronounced performance improvement, with MAE and RMSE reduced by 96.79\% and 96.92\% respectively, clearly demonstrating its superiority in enhancing control accuracy and stability of the system. As mentioned above, the poor performance of the standalone baseline controller is due to its inability to effectively address the lift loss inherent in the coaxial dual-rotor structure.
\begin{table}
    \caption{Comparison of Control Accuracy in $\boldsymbol e_z$ Direction}
    \renewcommand\arraystretch{1.5}
    \centering
    \begin{tabular}{cccc}
    \Xhline{1.0pt}
         \multicolumn{1}{c}{} & RMSE &MAE\\
         \Xhline{0.5pt}
         \multirow{1}{*}{Baseline controller}  & 0.7436 & 0.6388\\ 
         \multirow{1}{*}{NDO based control scheme}  & 0.0260 & 0.0227\\
         \multirow{1}{*}{\textbf{VDO based control scheme}}  & \textbf{0.0239} & \textbf{0.0197}\\ 
         \Xhline{1.0pt}
    \end{tabular}
    \label{VDO-effect}
\end{table}

For the reason that the VDO estimate $\Delta f$ instead of $\Delta {\boldsymbol F}^E$ of the drone, the conservativeness of conventional NDOs is reduced and the performance is improved. And the VDO can compensate the disturbance more timely than the baseline controller attributed to its active-disturbance-rejection characteristic inherited from the NDO \cite{DOBC}. It can be concluded that the VDO can estimate the VD disturbance accurately and timely by utilizing real-time state information of the drone.


\section{CONCLUSIONS}\label{conclusion}
This letter proposes an anti-disturbance control scheme for rotor drones subject to voltage drop of the battery. By adequately utilizing the state information and decoupling the disturbance from states of the drone, a VDO is designed to address the voltage drop issue of rotor drones conducting long-time fight or aggressive maneuvers. A rigorous mathematical analysis demonstrates the stability of the proposed control scheme, while real-world flight experiments verifies its effectiveness, adequately. 


\bibliographystyle{IEEEtran}
\bibliography{root}

\end{document}